\documentclass[11pt]{article}
\usepackage{fullpage}

\usepackage{url}
\usepackage[hidelinks]{hyperref}
\usepackage[utf8]{inputenc}
\usepackage[small]{caption}
\usepackage{graphicx}
\usepackage{amsmath}
\usepackage{mathtools}
\usepackage{booktabs}
\usepackage{comment}
\usepackage{amsthm}
\usepackage{amssymb}
\usepackage{bm}

\usepackage{cleveref}
\usepackage{color}
\usepackage{xspace}
\usepackage{tikz}
\usepackage{tabularx}
\usepackage{pbox}
\usepackage{framed}
\usepackage[shortlabels]{enumitem}
\usepackage{nicefrac}

\usepackage{natbib}

\usepackage{algorithm}
\usepackage[noend]{algorithmic}

\usepackage{xcolor}
\definecolor{Gred}{RGB}{219, 50, 54}
\definecolor{Ggreen}{RGB}{60, 186, 84}
\definecolor{Gblue}{RGB}{72, 133, 237}
\definecolor{Gyellow}{RGB}{247, 178, 16}
\definecolor{ToCgreen}{RGB}{0, 128, 0}
\definecolor{myGold}{RGB}{231,141,20}
\definecolor{myBlue}{rgb}{0.19,0.41,.65}
\definecolor{myPurple}{RGB}{175,0,124}

\usepackage[colorinlistoftodos,prependcaption,textsize=scriptsize]{todonotes}
\newcommand{\todoauthorsetup}[4]{%
	\expandafter\def\csname #1\endcsname##1{\todo[linecolor=#3,backgroundcolor=#3!25,bordercolor=#3]{#2: ##1}}%
}
\todoauthorsetup{pasin}{Pasin}{Gblue}{3}
\todoauthorsetup{adam}{Adam}{myPurple}{3}
\todoauthorsetup{ravi}{Ravi}{Ggreen}{3}
\todoauthorsetup{pritish}{Pritish}{orange}{3}

\newtheorem{theorem}{Theorem}
\newtheorem{corollary}{Corollary}

\newtheorem{lemma}{Lemma}

\theoremstyle{definition}
\newtheorem{definition}{Definition}

\newtheorem{observation}{Observation}

\newcommand{\bzero}{\mathbf{0}}
\newcommand{\bone}{\mathbf{1}}

\newcommand{\cX}{\mathcal{X}}
\newcommand{\cH}{\mathcal{H}}
\newcommand{\cF}{\mathcal{F}}

\DeclareMathOperator{\poly}{poly}

\newcommand{\oW}{\overline{W}}

\newcommand{\Paren}[1]{\left(#1\right)}

\newcommand{\etdprob}{\textsc{Extended-Threshold-Dimension}\xspace}
\newcommand{\gapetdprob}[2]{\textsc{Gap-Extended-Threshold-Dimension}$\left(#1,#2\right)$\xspace}
\newcommand{\tdprob}{\textsc{Threshold-Dimension}\xspace}

\newcommand{\exthd}[1]{\text{ExThD}\Paren{#1}}
\newcommand{\thd}[1]{\text{ThD}\Paren{#1}}
\newcommand{\sli}[1]{\text{SLI}\Paren{#1}}
\newcommand{\li}[1]{\text{LI}\Paren{#1}}
\newcommand{\thdclosure}[1]{\thd{\overline{#1}}}
\newcommand{\exactbicliqueprob}{\textsc{Maximum-Edge-Biclique}\xspace}
\newcommand{\exactsemiladderprob}{\textsc{Semi-Ladder-Index}\xspace}
\newcommand{\ladderprob}{\textsc{Ladder-Index}\xspace}
\newcommand{\gapbicliqueprob}[2]{\textsc{Gap-Balanced-Biclique}$(#1,#2)$\xspace}
\newcommand{\tgapbicliqueprob}{\textsc{Gap-Balanced-Biclique}\xspace}

\newcommand{\gapclosuretdprob}[2]{\textsc{Gap-(Closure-)Threshold-Dimension}$(#1,#2)$\xspace}
\newcommand{\gapsliprob}[2]{\textsc{Gap-(Semi-)Ladder-Index}$(#1,#2)$\xspace}
\newcommand{\tsetsplit}{\textsc{2-2-Set Splitting}\xspace}
\newcommand{\baltsetsplit}{\textsc{Balanced 2-2-Set Splitting}\xspace}

\newcommand{\oT}{\overline{T}}
\newcommand{\cHbase}{\cH_{\mathrm{base}}}
\newcommand{\supp}{\mathrm{supp}}
\newcommand{\argmin}{\mathrm{argmin}}

\allowdisplaybreaks

\title{On the Computational Complexity of\\ (Extended) Threshold Dimension and (Semi-)Ladder Index}
\author{Pasin Manurangsi \\ Google Research}
\date{\today}

\begin{document}

\maketitle

\begin{abstract}
We study the complexity of computing the \emph{Threshold dimension} of a hypothesis class and its variant, the \emph{Extended threshold dimension}. For the latter, we prove that it is both NP-hard and co-NP-hard, which (partially) answers an open question of Dmitriev et al. (SODA 2026). Furthermore, by relating the problem to a variant of Maximum Balanced Biclique, we prove strong hardness of approximation for both dimensions, including in the parameterized setting.

As an intermediate result, we also prove hardness (of approximation) results for computing the \emph{ladder index} and the \emph{semi-ladder index} (Fabianski et al., STACS 2019), which have recently been used in the design of fixed-parameter tractable algorithms.
\end{abstract}

\section{Introduction}

For a finite space $\cX$, a hypothesis class $\cH \subseteq \{0, 1\}^{\cX}$ is a set of  hypotheses $h: \cX \to \{0, 1\}$. Roughly speaking, its \emph{Threshold dimension} is the number of thresholds that can be formed from restrictions of $\cH$. This can be formalized as follows.

\begin{definition}[Threshold Dimension~\cite{shelah_classification_1978,hodges_shorter_1997}] \label{def:thd}
The \emph{Threshold dimension} of a hypothesis class $\cH \subseteq \{0, 1\}^{\cX}$, denoted by $\thd{\cH}$, is defined as the largest non-negative integer $d$ such that there exists $x_1, \dots, x_d \in \cX$ and $t_0, \dots, t_d \in \cH$ where $t_i(x_j) = 1$ if and only if $j \leq i$.

Such $x_1, \dots, x_d, t_0, \dots, t_d$ is referred to as the \emph{witness} of $\thd{\cH}$.
\end{definition}

Threshold dimension is closely related to the \emph{Littlestone dimension}, which characterizes the optimal mistake bound in online learning \cite{littlestone_learning_1988}. In particular, both dimensions are within an exponential factor of each other. More recently, the Threshold dimension has also been used for characterization of PAC learning with (approximate) differential privacy \cite{alon_private_2019}.

Dmitriev et al.~\cite{DmitrievFHS26} initiated a study of online learning \emph{with replays} where, in each round, the adversary can choose to reveal either the correct label or the label resulting from an output hypothesis from a previous round. Here mistakes are only counted in the former case. Remarkably, they show that the mistake bound can be characterized by a variant of Threshold dimension, called the \emph{Extended threshold dimension}. To define this, we recall the definition of closure and $f$-representation.

\begin{definition}[Closure]
For any hypothesis class $\cF \subseteq \{0, 1\}^{\cX}$, we write $\bigwedge_{f \in \cF} f$ to denote the hypothesis $g: \cX \to \{0, 1\}$ where $g(x) = \bigwedge_{f \in \cF} f(x)$ for all $x \in \cX$.

For any hypothesis class $\cH \subseteq \{0, 1\}^{\cX}$, its \emph{closure} $\overline{\cH}$ is the class $\{\bigwedge_{f \in \cF} f \mid \emptyset \ne \cF \subseteq \cH\}$.
\end{definition}

\begin{definition}[$f$-representation]
For any $f, h: \cX \to \{0, 1\}$, we write $f \oplus h$ to denote the hypothesis $g: \cX \to \{0, 1\}$ where $g(x) = f(x) \oplus h(x)$ for all $x \in \cX$.

For any $f: \cX \to \{0, 1\}$ and $\cH \subseteq \{0, 1\}^\cX$, the \emph{$f$-representation} of $\cH$, denoted by $\cH^{f}$, is the class $\{f \oplus h \mid h \in \cH\}$. We refer to $f$ as the \emph{shift} of $\cH^f$.
\end{definition}

The Extended threshold dimension can then be defined as follows:

\begin{definition}[Extended threshold Dimension~\cite{DmitrievFHS26}]
For any $\cH \subseteq \{0, 1\}^\cX$, its \emph{Extended threshold dimension} is
the minimum Threshold dimension over all possible $f$-representations of $\cH$ after taking the closure:
$$\exthd{\cH} = \min_{f: \cX \to \{0, 1\}} \thdclosure{\cH^f}.$$
\end{definition}

\subsection{Our Results}

While the computational complexity of other dimensions, such as VC and Littlestone dimensions, have been thoroughly investigated~\cite{Schaefer99,Schaefer00,MosselU02,PapadimitriouY96,FrancesL98,ManurangsiR17,Manurangsi23,NEURIPS2025_5f780c54}, the computational complexity of the Threshold dimension and its variants remains largely unexplored. In this work, we initiate the study on this topic and prove several hardness results.

\paragraph{Threshold Dimension.}
To formally study the computational hardness of computing Threshold dimension, we define the following decision problem\footnote{Throughout this work, we assume that the hypothesis class is encoded as a binary matrix $\{0, 1\}^{\cH \times \cX}$.}:
\begin{center}
\fbox{
\begin{minipage}{0.95\textwidth}
\textbf{Problem:} \tdprob \\
\textbf{Input:} A finite space $\mathcal{X}$, a hypothesis class $\mathcal{H} \subseteq \{0, 1\}^\mathcal{X}$, and a target integer $k \ge 0$. \\
\textbf{Question:} Is $\thd{\cH} \ge k$? %(i.e., does there exist an orientation shift $f \in \{0, 1\}^\mathcal{X}$ such that $\text{ThD}(\overline{\mathcal{H}^f}) \le k$?)
\end{minipage}
}
\end{center}

Note that this problem is clearly in NP, as the witness (in \Cref{def:thd}) can be efficiently verified. 
Our first result is to show the NP-hardness of this problem:
\begin{theorem} \label{thm:td-np-hard}
\tdprob is NP-complete.
\end{theorem}

In certain scenarios, it might be sufficient to \emph{approximately} compute the threshold dimension. Unfortunately, we show strong hardness of approximation results, under the Gap Exponential Time Hypothesis (Gap-ETH)\footnote{Gap-ETH~\cite{Dinur16,ManurangsiR17-gap-eth} postulates that, for some constant $\delta > 0$, there is no $2^{o(n)}$-time algorithm that can distinguish between a satisfiable 3CNF formula and one which is not even $(1 - \delta)$-satisfiable. Here $n$ denotes the number of variables. (This is a strengthening of ETH~\cite{ImpagliazzoP01,ImpagliazzoPZ01} which asserts this for $\delta = 0$.)}.

\begin{theorem} \label{thm:td-inapprox}
Assuming Gap-ETH, there is no polynomial-time $|\cX|^{o(1)}$-approximation algorithm or $|\cH|^{o(1)}$-approximation algorithm for Threshold dimension.
\end{theorem}

Note that there is a trivial linear-time $\min\{|\cX|, |\cH|\}$-approximation algorithm for $\thd{\cH}$, since it is always upper bounded by $\min\{|\cX|, |\cH|\}$. It remains an interesting question whether we can improve the inapproximability ratio to match this algorithm. (See \Cref{sec:conclusion_and_open_questions} for discussion.)

Another possible relaxation is through parameterized algorithms. Recall that, for a parameter $k$, fixed-parameter tractable (FPT) algorithms are those that run in $T(k) \cdot N^{O(1)}$ time where $T$ can be any function and $N$ denotes the problem size\footnote{We refer interested readers to \cite{DowneyF13} for further background on parameterized complexity.}. For our result, we parameterize by the optimum and we say that an algorithm is an $\alpha$-approximation\footnote{We refer interested readers to \cite{FeldmannSLM20} for a survey on FPT approximation algorithms and hardness results.} if it can distinguish between $\thd{\cH} \geq k$ and $\thd{\cH} \leq k / \alpha$. While it is trivial to achieve $O(k)$-approximation, we show that significantly improving upon this is unlikely, even for FPT algorithms:

\begin{theorem} \label{thm:td-fpt-inapprox}
Assuming Gap-ETH, there is no FPT $o(k)$-approximation algorithm for Threshold dimension.
\end{theorem}

\paragraph{Extended threshold Dimension.}
Similarly, we can define the decision problem for Extended threshold dimension as follows:
\begin{center}
\fbox{
\begin{minipage}{0.95\textwidth}
\textbf{Problem:} \etdprob \\
\textbf{Input:} A finite space $\mathcal{X}$, a hypothesis class $\mathcal{H} \subseteq \{0, 1\}^\mathcal{X}$, and a target integer $k \ge 0$. \\
\textbf{Question:} Is $\text{ExThD}(\mathcal{H}) \geq k$? %(i.e., does there exist an orientation shift $f \in \{0, 1\}^\mathcal{X}$ such that $\text{ThD}(\overline{\mathcal{H}^f}) \le k$?)
\end{minipage}
}
\end{center}

Note that, unlike \tdprob, it is a priori unclear if \etdprob belongs to NP. In particular, it is only straightforward to see that the problem belongs to $\Pi_2$: $\exthd{\cH} \geq k$ iff for all $f$, there exists a witness (see \Cref{lem:closure_Threshold} below) that $\thdclosure{\cH^f} \geq k$. Indeed, we show that it is unlikely to be in the class NP, as it is both NP-hard and co-NP-hard.

\begin{theorem} \label{thm:etd-np-hard}
\etdprob is NP-hard.
\end{theorem}
\begin{theorem} \label{thm:etd-co-np-hard}
\etdprob is co-NP-hard.
\end{theorem}

The above establishes computational barrier for computing Extended threshold dimension, which (partially) answers the question of \cite{DmitrievFHS26}. Similar to above, we also provide hardness of approximation results for Extended threshold dimensions.

\begin{theorem} \label{thm:etd-inapprox}
Assuming Gap-ETH, there is no polynomial-time $|\cX|^{o(1)}$-approximation algorithm or $|\cH|^{o(1)}$-approximation algorithm for Extended threshold dimension.
\end{theorem}

\begin{theorem} \label{thm:etd-fpt-inapprox}
Assuming Gap-ETH, there is no FPT $o(k)$-approximation algorithm for Extended threshold dimension.
\end{theorem}

\paragraph{Ladder and Semi-Ladder Indices.}
Perhaps interestingly, our hardness results are shown via viewing the problems from a graph-theoretic perspective. In particular, it turns out that Threshold dimension of a class $\cH$ and its closure $\overline{\cH}$ are closely related to the notion of Ladder and Semi-Ladder indices from the parameterized algorithm literature \cite{FabianskiPST19-semiladder}, which we define below.

\begin{definition}[Ladder Index~\cite{FabianskiPST19-semiladder}]
A \emph{ladder index} of a bipartite graph $G = (L, R, E)$, denoted by $\li{G}$, is the largest non-negative integer $\ell$ such that there exists $a_1, \dots, a_\ell \in L$ and $b_1, \dots ,b_\ell \in R$ such that, for all $i, j \in [\ell]$, $(a_i, b_j) \in E$ if and only if $i > j$.

We refer to such $a_1, \dots, a_\ell, b_1, \dots, b_\ell$ as a \emph{ladder of order $\ell$ of $G$}. 
\end{definition}

\begin{definition}[Semi-Ladder Index~\cite{FabianskiPST19-semiladder}]
A \emph{semi-ladder index} of a bipartite graph $G = (L, R, E)$, denoted by $\sli{G}$, is the largest non-negative integer $\ell$ such that there exists $a_1, \dots, a_\ell \in L$ and $b_1, \dots ,b_\ell \in R$ that satisfy the following:
\begin{itemize}
\item $(a_i, b_i) \notin E$ for all $i \in [\ell]$, and,
\item $(a_i, b_j) \in E$ for all $i, j \in [\ell]$ such that $i > j$.
\end{itemize}
We refer to such $a_1, \dots, a_\ell, b_1, \dots, b_\ell$ as a \emph{semi-ladder of order $\ell$ of $G$}. 
\end{definition}

We note that, for any bipartite graph $G$, $\sli{G} \geq \li{G}$. This is simply because any ladder of order $\ell$ is also a semi-ladder of order $\ell$. However, the inverse does not hold since there might be an edge between $(a_i, b_j)$ for some $i < j$ in a semi-ladder. 

We can now define the associated computational problems as follows.

\begin{center}
\fbox{
\begin{minipage}{0.95\textwidth}
\textbf{Problem:} \exactsemiladderprob \\
\textbf{Input:} A bipartite graph $G = (L, R, E)$ and a target integer $k \geq 0$. \\
\textbf{Question:} Is $\sli{G} \geq k$?
\end{minipage}
}
\end{center}

\begin{center}
\fbox{
\begin{minipage}{0.95\textwidth}
\textbf{Problem:} \ladderprob \\
\textbf{Input:} A bipartite graph $G = (L, R, E)$ and a target integer $k \geq 0$. \\
\textbf{Question:} Is $\li{G} \geq k$?
\end{minipage}
}
\end{center}

It is obvious that these problems are in NP. We show that they are NP-hard:

\begin{theorem} \label{thm:indices-np-hard-intro}
Both \exactsemiladderprob and \ladderprob are NP-complete.
\end{theorem}

Finally, similar to above, we also provide hardness of approximation results for these indices.

\begin{theorem} \label{thm:ladder-inapprox}
Assuming Gap-ETH, there is no polynomial-time $n^{o(1)}$-approximation algorithm for Ladder index or Semi-ladder index, where $n$ denotes the number of vertices in the input graph.
\end{theorem}

\begin{theorem} \label{thm:ladder-fpt-inapprox}
Assuming Gap-ETH, there is no FPT $o(k)$-approximation algorithm for Ladder index or Semi-ladder index.
\end{theorem}
\section{Preliminaries}

For any $S \subseteq \cX$, we use $\bone_S$ to denote the hypothesis that is the indicator of $S$, i.e. $\bone_S(x) = 1$ iff $x \in S$. For $x \in \cX$, we write $\bone_x$ and $\bone$ as abbreviations for $\bone_{\{x\}}$ and $\bone_{\cX}$, respectively. We refer to $\bone_{\{x\}}$ as a \emph{singleton} and $\bone_{\cX \setminus \{x\}}$ as a \emph{co-singleton}. For every $f: \cX \to \{0, 1\}$, we write $\supp(f)$ to denote its support, i.e. $f^{-1}(1)$, and use $\|f\|$ as a shorthand for $|\supp(f)|$.

Recall that a biclique is simply a (not necessarily balanced) complete bipartite graph. A balanced biclique is one whose two sides have the same number of vertices. The size of a balanced biclique is the number of vertices on each side.

\subsection{Promise Problems}

To prove hardness (of approximation) results, it is often useful to consider \emph{promise problems}. A promise problem $\Pi$ is a tuple of languages $(\Pi_{YES}, \Pi_{NO})$ such that $\Pi_{YES} \cap \Pi_{NO} = \emptyset$. A (deterministic) algorithm $A$ is said to solve a promise problem if $A(x) = 1$ for all $x \in \Pi_{YES}$, and $A(x) = 0$ for all $x \in \Pi_{NO}$. Note that there is no requirement for $x \notin \Pi_{YES} \cup \Pi_{NO}$. 

Reductions and NP-hardness of promise problems can be defined analogously to standard decision problems. (See~\cite{Goldreich06a} for more detailed discussion.)

\subsection{Useful Properties of Threshold Dimensions (and Its Variants)}

Below we list a few properties of Threshold dimension and its variants, which will be subsequently useful in our proofs.

\subsubsection{Structural Characterization of Thresholds in Closures}

We begin by establishing a lemma characterizing when the closure $\overline{\cH}$ has Threshold dimension $d$. We stress that the generators $g_1, \dots, g_{d + 1}$ in the lemma below are from the base hypothesis class $\cH$ (not the closure $\overline{\cH}$), which makes this lemma convenient for our subsequent proofs.

\begin{lemma} \label{lem:closure_Threshold}
Let $\cH \subseteq \{0,1\}^\cX$. We have $\thdclosure{\cH} \ge d$ if and only if there exist $x_1, \dots, x_d \in \cX$ and $g_1, \dots, g_d, g_{d + 1} \in \cH$ such that:
\begin{enumerate}
    \item $g_{d + 1}(x_j) = 1$ for all $j \in [d]$.
    \item For all $i \in [d]$, $g_i(x_i) = 0$, and $g_i(x_j) = 1$ for all $j < i$.
\end{enumerate}
We refer to $g_1, \dots, g_{d + 1}$ as the \emph{generators} for $\thdclosure{\cH} \geq d$.
\end{lemma}
\begin{proof}
($\Rightarrow$) Assume $\thdclosure{\cH} \ge d$. There exist $x_1, \dots, x_d$ and $t_0, t_1, \dots, t_d \in \overline{\cH}$ with $t_i(x_j) = 1$ iff $j \le i$. 
Since $t_d(x_j) = 1$ for all $j \in [d]$, and $t_d = \bigwedge_{g \in A_d} g$ for some $\emptyset \neq A_d \subseteq \cH$, any $g_{d + 1} \in A_d$ must satisfy $g_{d + 1}(x_j) = 1$ for all $j \in [d]$.
For each $i \in [d]$, since $t_{i - 1}(x_i) = 0$ and $t_{i-1} = \bigwedge_{g \in A_{i-1}} g$ for some $\emptyset \neq A_{i-1} \subseteq \cH$, there is some $g_i \in A_{i-1}$ such that $g_i(x_i) = 0$. Furthermore, because $t_{i-1}(x_j) = 1$ for all $j < i$, every element in $A_{i-1}$, including $g_i$, must evaluate to $1$ on $x_1, \dots, x_{i-1}$. This yields the desired generators $g_1, \dots, g_{d + 1}$.

($\Leftarrow$) Assume such $x_1, \dots, x_d$ and $g_1, \dots, g_{d+1} \in \cH$ exist. Define $t_i =  \bigwedge_{k=i+1}^{d+1} g_k$ for $i \in \{0, \dots, d\}$. It is simple to see that $t_i(x_j) = 1$ iff $j \le i$, meaning that $\thdclosure{\cH} \ge d$.
\end{proof}

\subsubsection{Restrictions, Subclasses and Threshold Dimensions}

We start by defining the restriction of a subclass onto a subdomain:

\begin{definition}[Restriction]
For any $\cH \subseteq \{0, 1\}^{\cX}$ and $\cX' \subseteq \cX$, the restriction of $\cH$ onto $\cX'$ is defined as $\cH|_{\cX'} := \{h|_{\cX'} \mid h \in \cH\} \subseteq \{0,1\}^{\cX'}$.
\end{definition}

The following two observations regarding monotonicity of threshold dimension (for subclasses and restrictions) are obvious.

\begin{observation} \label{obs:superset}
For any $\cH' \subseteq \cH \subseteq \{0, 1\}^{\cX}$, we have $\thd{\cH} \geq \thd{\cH'}$.
\end{observation}

\begin{observation} \label{obs:restriction-doesnot-decrease}
For any $\cH \subseteq \{0, 1\}^{\cX}$ and $\cX' \subseteq \cX$, we have $\thd{\cH} \geq \thd{\cH|_{\cX'}}$.
\end{observation}

In some cases, the above inequality can become an equality. Below, we list a couple such cases. First is when we remove a ``constant'' element:

\begin{observation} \label{obs:constant-removal}
For any $\cH \subseteq \{0, 1\}^{\cX}$, an element $x \in \cX$ is said to be \emph{constant} if $h(x)$ has the same value for all $h \in \cH$. For such an element $x$, we have $\thd{\cH|_{\cX \setminus \{x\}}} = \thd{\cH}$, $\thdclosure{\cH|_{\cX \setminus \{x\}}} = \thdclosure{\cH}$ 
and $\exthd{\cH|_{\cX \setminus \{x\}}} = \exthd{\cH}$.
\end{observation}

\begin{proof}
To see that $\thd{\cH|_{\cX \setminus \{x\}}} = \thd{\cH}$, observe that the coordinates $x_1, \dots, x_d$ that witness $\thd{\cH} = d$ cannot be constant coordinates. Thus, removing this constant coordinate does not affect the threshold dimension. 

As for the remaining two equalities, since $x$ is constant in $\cH$, $x$ is also constant in $\cH^f$ for any $f: \cX \to \{0, 1\}$. In turn, this implies that $x$ is constant in $\overline{\cH^f}$. Thus, we have $\thdclosure{\cH^f} = \thdclosure{\cH^f|_{\cX \setminus \{x\}}}$. This indeed implies that $\exthd{\cH|_{\cX \setminus \{x\}}} = \exthd{\cH}$.
\end{proof}

Another equality case is when two elements are ``twins'':

\begin{observation} \label{obs:twin-removal}
For any $\cH \subseteq \{0, 1\}^{\cX}$, two distinct elements $x, x' \in \cX$ are \emph{twins} if $h(x) = h(x')$ for all $h \in \cH$. For any twins $x, x'$, $\thd{\cH|_{\cX \setminus \{x\}}} = \thd{\cH}$, and $\thdclosure{\cH|_{\cX \setminus \{x\}}} = \thdclosure{\cH}$.
\end{observation}

\begin{proof}
The first claim simply follows from the fact that $x, x'$ cannot be used together in a witness for Threshold dimension of $\cH$ and, if any witness uses $x$, it can be replaced by $x'$. The second claim then follows from the fact that $x, x'$ remain twins in the closure $\overline{\cH}$.
\end{proof}

%Finally, we also remark that the deficit in the dimension can also be bounded, e.g. by the number of elements removed.

%\begin{observation} \label{obs:nonconstant-removal}
%For any $\cH \subseteq \{0, 1\}^{\cX}$ and any $\cX' \subseteq \cX$, we have  $\thd{\cH|_{\cX'}} \geq \thd{\cH} - |\cX \setminus \cX'|$, %$\exthd{\cH|_{\cX \setminus \{x\}}} = \exthd{\cH}$ 
%and $\exthd{\cH|_{\cX'}} \geq \exthd{\cH} - |\cX \setminus \cX'|$.
%\end{observation}

%\begin{proof}
%Suppose $\thd{\cH} = d$, and let $t_0, \dots, t_d \in \cH$ and $x_1, \dots, x_d \in \cX$ be the witness. Let $d' = d - |\cX \setminus \cX'|$, and $i_1 < \dots < i_{d'}$ be indices such that $x_{i_1}, \dots, x_{i_{d'}} \notin \cX'$. It is simple to see that $t_0, t_{i_1}, \dots, t_{i_{d'}}$ and $x_{i_1}, \dots, x_{i_{d'}}$ is a witness that $\thd{\cH|_{\cX'}} \geq d'$ as desired.

%The bound on Extended threshold dimension follows immediately by applying the previous bound on $\cH^f$ for all $f: \cX \to \{0, 1\}$.
%\end{proof}

\section{NP-Hardness Results}

In this section, we will prove our NP-hardness results.

\subsection{Ladder Index and Semi-Ladder Index}

We begin with Ladder and Semi-Ladder Indices. For convenience, let us define the following (promise) problem. Note that this promise problem has the Ladder index in the YES case but the Semi-ladder index in the NO case.

\begin{center}
\fbox{
\begin{minipage}{0.95\textwidth}
\textbf{Problem:} \gapsliprob{k_1}{k_2} \\
\textbf{Input:} A bipartite graph $G = (L, R, E)$ and positive integers $k_1 \geq k_2$. \\
\textbf{YES Case:} $\li{G} \geq k_1$. \\
\textbf{NO Case:} $\sli{G} < k_2$. 
\end{minipage}
}
\end{center}

In this section, we will only use the exact (i.e. ``non-gap'') version of the problem where $k_1 = k_2 = \ell$. However, in the next section, we will also use the gap version as well.
Our main result here is that the exact version of this problem is NP-hard:
\begin{theorem} \label{thm:gapsli-np-hard}
\gapsliprob{\ell}{\ell} is NP-hard.
\end{theorem}

Note that \Cref{thm:gapsli-np-hard} immediately implies \Cref{thm:indices-np-hard-intro} since $\sli{G} \geq \li{G}$.

To prove \Cref{thm:gapsli-np-hard}, we will reduce from yet another promise problem. To state this problem precisely, we need an additional notation: For every graph $G = (V, E)$, we let $B[G^o]$ denote the bipartite graph $H = (L, R, E_H)$ where
\begin{itemize}
\item $L, R$ are copies of $V$. Namely, let $L = V \times \{1\}$ and $R = V \times \{2\}$. %For each $v_j \in V$, we let $a_j$ and $b_j$ denote the copies of $v_j$ in $L$ and $R$, respectively.
\item There exists an edge $((u, 1), (v, 2)) \in E_H$ if and only if\footnote{This is the same as the \emph{bipartite double cover} of $G$, except that $B[G^o]$ contains edges $((v, 1), (v,2))$ for all $v \in V$.} $u = v$ or $(u, v) \in E$.
\end{itemize}

The promise problem we reduce from is the following, which is essentially the same as the standard \emph{Maximum Edge Biclique} problem on $B[G^o]$, except that the YES case is stronger.

\begin{center}
\fbox{
\begin{minipage}{0.95\textwidth}
\textbf{Problem:} \exactbicliqueprob \\
\textbf{Input:} A graph $G = (V, E)$ and a positive integer $k$. \\
\textbf{YES Case:} $G$ contains a $k$-clique. \\
\textbf{NO Case:} Any biclique in $B[G^o]$ has at most $k^2 - 1$ edges.
\end{minipage}
}
\end{center}

While the NP-hardness for this version of the problem does not follow from the classical reduction for the standard Maximum Edge Biclique problem (e.g. from \cite{Peeters03}), it is simple to see that it follows from the reduction for Densest $k$-Subgraph in \cite{Manurangsi17} (with slightly different parameters). This gives the following:  

\begin{lemma} \label{lem:np-exactbiclique}
\exactbicliqueprob is NP-hard.
\end{lemma}

For completeness, we give the full proof of \Cref{lem:np-exactbiclique} in Appendix~\ref{app:biclique}. We are now ready to prove \Cref{thm:gapsli-np-hard}.

\begin{proof}[Proof of \Cref{thm:gapsli-np-hard}]
%The containment in NP follows immediately from the definition of semi-ladder index, since $a_1, \dots, a_\ell, b_1, \dots, b_\ell$ serve as the witness that can be checked in polynomial time.
%
%To prove NP-hardness of \exactsemiladderprob, 
We will reduce from \exactbicliqueprob. Given an instance $(G = (V, E), k)$ of \exactbicliqueprob where $V = \{v_1, \dots, v_n\}$, we create an instance $(G' = (L', R', E'), \ell)$ for \gapsliprob{\ell}{\ell} as follows.
\begin{itemize}
\item Let $t = 2k^2$ and $\ell = kt$.
\item Let each of $L'$ and $R'$ contain $t$ copies of $V$. We use $a^{(i)}_j$ (resp. $b^{(i)}_j$) to denote the $i$-th copy of $v_j \in V$ in $L$ (resp. $R$). Let $L^{(i)}$ (resp. $R^{(i)}$) denote the set $\{a^{(i)}_j \mid j \in [n]\}$ (resp. $\{b^{(i)}_j \mid j \in [n]\}$).
\item There is an edge between $(a^{(i)}_j, b^{(i')}_{j'})$ iff both of the following conditions hold:
\begin{itemize}
    \item $(i, j)$ is lexicographically (strictly) larger than $(i', j')$, and,
    \item  $j = j'$ or $(v_j, v_{j'}) \in E$.
\end{itemize}
\end{itemize}
The reduction runs in polynomial time. We next prove its completeness and soundness.

\paragraph{(Completeness)} Suppose that $v_{i_1}, \dots, v_{i_k}$ form a clique in $G$ where $i_1 < \cdots < i_k$. It is simple to see that $a^{(1)}_{i_1}, a^{(1)}_{i_2}, \dots, a^{(t)}_{i_k}$ and $b^{(1)}_{i_1}, b^{(1)}_{i_2}, \dots, b^{(t)}_{i_k}$ form a ladder of order $\ell = kt$ in $G'$.

\paragraph{(Soundness)} Suppose contrapositively that $G'$ contains a semi-ladder of order $\ell$: $u_1, \dots, u_\ell$, $w_1, \dots, w_\ell$. Let $A' = \{u_{\ell/2 + 1}, \dots, u_{\ell}\}$ and $B' = \{w_1, \dots, w_{\ell/2}\}$. By definition of a semi-ladder, $A', B'$ forms a biclique in $G'$. Then, let $A$ (resp. $B$) denote the set of vertices in $V$ such that at least one of its copies appear in $A'$ (resp. $B'$). By how $G'$ is constructed, we also have that $A \times \{1\}, B \times \{2\}$ form a biclique in $B[G^{o}]$. Now, let $i^A_{\min} := \min\{i \in [t] \mid A' \cap L^{(i)} \ne \emptyset\}$ and $i^B_{\max} := \max\{i \in [t] \mid B' \cap R^{(i)} \ne \emptyset\}$. Since $A', B'$ form a biclique, we must have that $i^A_{\min} \geq i^B_{\max}$. By definition, we also have $A' \subseteq \{a^{(i)}_j \mid v_j \in A, i \geq i^A_{\min}\}$ and $B' \subseteq \{b^{(i)}_j \mid v_j \in B, i \leq i^B_{\max}\}$. Thus,
\begin{align*}
\Paren{\ell/2} \cdot \Paren{\ell/2} = |A'| \cdot |B'| &\leq \Paren{\Paren{t - i^A_{\min} + 1} \cdot |A|} \cdot \Paren{\Paren{i^B_{\max}} \cdot |B|} \\
&\leq \Paren{\frac{t - i^A_{\min} + 1 + i^B_{\max}}{2}}^2 \cdot |A| \cdot |B| \\
&\leq \Paren{\frac{t + 1}{2}}^2 \cdot |A| \cdot |B|,
\end{align*}
where the second inequality uses the A.M.-G.M. inequality and the third follows from $i^A_{\min} \geq i^B_{\max}$.

Rearranging this, we have
\begin{align*}
|A| \cdot |B| \geq \frac{\ell^2}{(t + 1)^2} = \frac{k^2 t^2}{(t + 1)^2} = k^2 - \frac{(2t + 1)k^2}{(t + 1)^2} > k^2 - 1,
\end{align*}
where the last inequality is due to our choice of $t$. This means that $A \times \{1\}, B \times \{2\}$ form a biclique in $G$ with at least $k^2$ edges, as desired.
\end{proof}

\subsection{From (Semi-)Ladder Index to (Closure) Threshold Dimension}

It turns out that (Semi-)Ladder index is closely related to the Threshold dimension (of the closure). Namely, there is a simple reduction that turns a graph into a hypothesis class while turning the index to the dimension. To state this reduction, it is helpful to introduce another promise problem.

\begin{center}
\fbox{
\begin{minipage}{0.95\textwidth}
\textbf{Problem:} \gapclosuretdprob{k_1}{k_2} \\
\textbf{Input:} A finite space $\mathcal{X}$, a hypothesis class $\mathcal{H} \subseteq \{0, 1\}^\mathcal{X}$, and positive integers $k_1 \geq k_2$. \\
\textbf{YES Case:} $\thd{\cH} \geq k_1$ \\
\textbf{NO Case:} $\thdclosure{\cH} < k_2$ 
\end{minipage}
}
\end{center}

\begin{lemma} \label{lem:red-sli-to-closuretd}
%There is a polynomial-time reduction that takes in a bipartite graph $G = (L, R, E)$ and produces a hypothesis class $\cH \subseteq \{0, 1\}^{\cX}$ such that $\thdclosure{\cH} = \sli{G}$, and, 
There is a polynomial-time reduction from \gapsliprob{k_1}{k_2} to \gapclosuretdprob{k_1}{k_2}.
%Furthermore, if $G = (L, R, E)$ is the input instance and $\cH \in \{0, 1\}^{\cX}$  is the output instance, then $|\cH| \leq |L|+1, |\cX| \leq |R|$.
\end{lemma}

\begin{proof}
Let $\cX = R$, and let $\cH$ contain the following hypotheses:
\begin{itemize}
\item For every $u \in L$, the hypothesis $h_u$ where $h_u(v) = \bone[(u, v) \in E]$.
\item The all-one hypothesis $\bone$.
\end{itemize}
We next prove the completeness and soundness of the reduction.

\paragraph{(Completeness)} Suppose that $\li{G} \geq k_1$. That is, there exist $a_1, \dots, a_{k_1} \in L, b_1, \dots, b_{k_1} \in R$ that form a ladder in $G$. Consider the elements $b_1, \dots, b_{k_1} \in \cX$ and hypotheses $\bone, h_{a_1}, \dots, h_{a_{k_1}}$. It is simple to check that these form a witness that $\thd{\cH} \geq {k_1}$.

\paragraph{(Soundness)} Suppose that $\thdclosure{\cH} = d$. By \Cref{lem:closure_Threshold}, there exist $x_1, \dots, x_d \in \cX$ and $g_1, \dots, g_{d + 1} \in \cH$ such that $g_{d+1}(x_j) = 1$ for all $j \in [d]$, and $g_i(x_i) = 0$ and $g_i(x_j) = 1$ for all $i, j \in [d]$ with $j < i$. Since $g_i(x_i) = 0$ for all $i \in [d]$, we have $g_1, \dots, g_d \ne \bone$. Thus, for every $i \in [d]$, $g_i = h_{u_i}$ for some $u_i \in L$. This implies that $u_1, \dots, u_d, x_1, \dots, x_d$ form a semi-ladder of order $d$.
\end{proof}

We remark that, if we do not add the all-one hypothesis to $\cH$ in the above reduction, then we will instead have the inequalities $\li{G} - 1 \leq \thd{\cH} \leq \li{G}$ and $\sli{G} - 1 \leq \thdclosure{\cH} \leq \sli{G}$. While this is sufficient for hardness of approximation reductions (where we start from the gap version of \gapsliprob{k_1}{k_2}), it is insufficient for exact reductions.

Since $\thd{\cH} \leq \thdclosure{\cH}$, combining the above reduction (\Cref{lem:red-sli-to-closuretd}) with \Cref{thm:gapsli-np-hard} immediately yield the NP-hardness of Threshold dimension (\Cref{thm:td-np-hard}).

\subsection{From Closure Threshold Dimension to Extended Threshold Dimension}

We have also almost established the NP-hardness of Extended threshold dimension (\Cref{thm:etd-np-hard}), except for one crucial detail: $\exthd{\cH}$ can be smaller than $\thdclosure{\cH}$. We handle this in the next lemma, which shows that we can, in fact, use a simple reduction to ensure that $\exthd{\cH}$ is the same as $\thdclosure{\cH}$. To state the lemma, it is helpful to define yet another promise problem.

\begin{center}
\fbox{
\begin{minipage}{0.95\textwidth}
\textbf{Problem:} \gapetdprob{k_1}{k_2} \\
\textbf{Input:} A finite space $\mathcal{X}$, a hypothesis class $\mathcal{H} \subseteq \{0, 1\}^\mathcal{X}$, and positive integers $k_1 \geq k_2$. \\
\textbf{YES Case:} $\exthd{\cH} \geq k_1$ \\
\textbf{NO Case:} $\exthd{\cH} < k_2$ 
\end{minipage}
}
\end{center}

\begin{lemma} \label{lem:tdclosure-to-etd}
There is a polynomial-time reduction from \gapclosuretdprob{k_1}{k_2} to \gapetdprob{k_1}{k_2}.
\end{lemma}

Before we prove this lemma, let us note that, by applying \Cref{lem:red-sli-to-closuretd} and \Cref{lem:tdclosure-to-etd} to \Cref{thm:gapsli-np-hard}, we have established \Cref{thm:etd-np-hard}.

\begin{proof}[Proof of \Cref{lem:tdclosure-to-etd}]
Assume w.l.o.g. that $\thd{\cH} \geq 2$ and that $\cX$ does not have constant or twin coordinates (w.r.t. $\cH$); otherwise, from Observations~\ref{obs:constant-removal} and~\ref{obs:twin-removal}, we may simply remove them. Let $t = |\cX|$.
%First, let $\cX_{nc} \subseteq \cX$ denote the set of non-constant coordinates with respect to $\cH$. Let $\cH_{nc} = \cH|_{\cX_{nc}}$. By \Cref{obs:constant-removal}, we have $\exthd{\cH_{nc}} = \exthd{\cH}$. 
We construct $\cX', \cH'$ as follows.
\begin{itemize}
\item Let $\cX' = \cX^{(1)} \cup \cdots \cup \cX^{(t)}$ where each $\cX^{(i)}$ denotes a copy of $\cX$. We use $x^{(i)}$ to denote the copy of $x \in \cX$ in $\cX^{(i)}$.
\item The class $\cH'$ consists of the following hypotheses:
\begin{itemize}
\item The all zero hypothesis $\bzero$.
\item The singleton hypotheses $\bone_{x'}$ for all $x' \in \cX'$.
\item For each $h \in \cH$ and $i \in [t]$, we create a hypothesis $h^{(i)}$ where
\begin{align*}
h^{(i)}(x^{(j)}) =
\begin{cases}
h(x) &\text{ if } j = i, \\
0 &\text{ otherwise.}
\end{cases}
\end{align*}
We refer to $h^{(i)}$ as the $i$-th copy of $h \in \cH$.
\end{itemize}
\end{itemize}
It is clear that the reduction runs in $\poly(|\cH|, |\cX|, t)$ time, and that $|\cX'| \leq t \cdot |\cX|$ and $|\cH'| \leq t \cdot (|\cH| + |\cX|) + 1$. We will next prove the completeness and soundness of the reduction.

\paragraph{(Completeness)} 
We will prove that, if $\thd{\cH} \geq k_1$, then $\exthd{\cH'} \geq k_1$.
To see that this is the case, consider any $f: {\cX'} \to \{0, 1\}$. We consider two cases:
\begin{itemize}
\item Case I: $\|f\| \geq k_1$. In this case, $f \oplus \bone_{x}$ for all $x \in \supp(f)$ become co-singletons when restricted to $\supp(f)$. From this and from $f = f \oplus \bzero \in (\cH')^f$, we have $\thdclosure{(\cH')^f} \geq \|f\| \geq k_1$.
\item Case II: $\|f\| < k_1$. Since $t \geq k_1$, there is $i \in [t]$ with $\supp(f) \cap \cX^{(i)} = \emptyset$. Thus, we have
\begin{align*}
\thdclosure{(\cH')^f} \geq \thdclosure{(\cH')^f|_{\cX^{(i)}}}
=  \thdclosure{\cH'|_{\cX^{(i)}}} \geq k_1,
\end{align*}
where the inequalities follow from Observations~\ref{obs:restriction-doesnot-decrease} and \ref{obs:superset}, respectively.
\end{itemize}
Thus, in both cases, we have $\thdclosure{(\cH')^f} \geq k_1$. This implies $\exthd{\cH'} \geq k_1$ as desired.

\paragraph{(Soundness)}
%\paragraph{Proof of $\exthd{\cH'} \leq \thdclosure{\cH}$.} 
We would like to show that, if $\thdclosure{\cH} < k_2$, then $\exthd{\cH'} < k_2$.
In fact, we will show an even stronger result that $\thdclosure{\cH'} \leq \thdclosure{\cH}$. The statement is trivial if $\thdclosure{\cH'} = 1$. Suppose that $\thdclosure{\cH'} = d \geq 2$ where $s'_1, \dots, s'_{d + 1} \in \cH'$ and $x'_1, \dots, x'_d \in \cX'$ are the witness (according to \Cref{lem:closure_Threshold}). First, notice that $x'_1, \dots, x'_d$ must be from the same copy, i.e. $x'_1, \dots, x'_d \in \cX^{(i)}$ for some $i \in [t]$. This is simply because there is no hypothesis that assigns 1 to coordinates from different copies. Let $x_1, \dots, x_d \in \cX$ denote the corresponding elements to $x'_1, \dots, x'_d$ from the underlying base space $\cX$; that is, $x'_1 = x_1^{(i)}, \dots, x'_d = x_d^{(i)}$.

We claim that we may w.l.o.g. take $s'_1, \dots, s'_{d + 1}$ to be the $i$-th copy of some hypotheses from $\cH$. We note that $s'_3, \dots, s'_{d + 1}$ must assign 1 to at least two coordinates in $\cX^{(i)}$; thus, they must be $i$-th copy of some hypotheses from $\cH$. As for $s'_1, s'_2$, we may select them as follows. Since $x'_1, x'_2$ are not twins in $\cH$, there exists $h \in \cH$ such that $h(x_1) \ne h(x_2)$. Consider two cases: 
\begin{itemize}
\item Case I: $h(x_1) = 1$. Since $x_1$ is not a constant coordinate, there exists another hypothesis $h'$ such that $h'(x_1) = 0$. Thus, we may take $s'_2 = h^{(i)}$ and $s'_1 = (h')^{(i)}$ respectively. 
\item Case II: $h(x_1) = 0$. In this case, we may swap $x'_1, x'_2$ and use the same argument as above.
\end{itemize}

Thus, we may assume that $s'_1 = s^{(i)}_1, \dots, s'_{d+1} = s^{(i)}_{d + 1}$ for some $s_1, \dots, s_{d + 1} \in \cH$. This implies that $s_1, \dots, s_{d+1}$ and $x_1, \dots, x_d$ are witness (according to \Cref{lem:closure_Threshold}) for $\thdclosure{\cH}$. This means that $\thdclosure{\cH} \geq d = \thdclosure{\cH'}$ as claimed.
\end{proof}
\section{Hardness of Approximation}

In this section, we will prove hardness of approximation for the problems of interest. We will reduce from the hardness of approximation results of the Maximum Balanced Biclique problem. We will again use the gap version of the problem, as stated below. 

\begin{center}
\fbox{
\begin{minipage}{0.95\textwidth}
\textbf{Problem:} \gapbicliqueprob{k_1}{k_2} \\
\textbf{Input:} A graph $G = (V, E)$ and positive integers $k_1 \geq k_2$ \\
\textbf{YES Case:} $G$ contains a clique of size $k_1$ \\
\textbf{NO Case:} $B[G^o]$ does \emph{not} contain a balanced biclique of size $k_2$ 
\end{minipage}
}
\end{center}

We remark that, similar to \exactbicliqueprob in the previous section, this gap version is slightly stronger as it requires that the graph is $B[G^o]$ and that the biclique in the YES case is form from the $k$-clique of the underlying graph $G$. Indeed, as we explain below, not all hardness for Maximum Balanced Biclique can be written in this form.

We observe that the same reduction as in \Cref{thm:gapsli-np-hard} (but without any repetition) already give a reduction from this problem to the \gapsliprob{\cdot}{\cdot}, except with a loss of a factor of 2 in the gap. We remark that this loss is exactly why repetition was needed in the previous section, since we start off without any gap.

\begin{lemma} \label{lem:biclique-to-sli}
There is a polynomial-time reduction from \gapbicliqueprob{k_1}{k_2} to \gapsliprob{k_1}{2k_2}.
\end{lemma}

\begin{proof}
Given an instance $G = (V, E)$ of \gapbicliqueprob{k_1}{k_2} where $V = \{v_1, \dots, v_n\}$, we create an instance $G' = (L', R', E')$ for \gapsliprob{k_1}{2k_2} as follows.
\begin{itemize}
\item Let each of $L'$ and $R'$ be a copy of $V$. We use $a_j$ (resp. $b_j$) to denote the copy of $v_j \in V$ in $L$ (resp. $R$). 
\item There is an edge between $(a_j, b_{j'})$ iff $j > j'$ and $(v_j, v_{j'}) \in E$.
\end{itemize}
It is clear that the reduction runs in polynomial time.

\paragraph{(Completeness)} Suppose that $v_{i_1}, \dots, v_{i_{k_1}}$ forms a clique in $G$ where $i_1 < \cdots < i_{k_1}$. Then, $a_{i_1}, a_{i_2}, \dots, a_{i_{k_1}}$ and $b_{i_1}, b_{i_2}, \dots, b_{i_{k_1}}$ form a ladder of order $k_1$ in $G'$.

\paragraph{(Soundness)} Suppose contrapositively that $G'$ contains a semi-ladder of order $2k_2$: $a_{i_1}, \dots, a_{i_{2k_2}},$ $b_{j_1}, \dots, b_{j_{2k_2}}$. Let $A' = \{a_{i_{k_2 + 1}}, \dots, a_{i_{2k_2}}\}$ and $B' = \{b_{j_1}, \dots, b_{j_{k_2}}\}$; they form a biclique in $G'$. Thus, $\{(v_{i_{k_2+1}},1), \dots, (v_{i_{2k_2}}, 1)\}$ and $\{(v_{j_{1}},2), \dots, (v_{j_{k_2}}, 2)\}$ form a biclique in $B[G^o]$, as desired.
\end{proof}

While Maximum Balanced Biclique is \emph{not} known to be NP-hard to approximate, hardness of approximation results are known under other (stronger) assumptions~\cite{Fei02,FK04,BhangaleGHKK16,Khot06,Manurangsi17,Manurangsi17-icalp,ChalermsookCKLM20,ManurangsiRS21}. Unfortunately, some of these constructions (e.g. \cite{Khot06,BhangaleGHKK16,Manurangsi17-icalp}) are not of the form \gapbicliqueprob{\cdot}{\cdot}. Nevertheless, some of the others~\cite{Fei02,Manurangsi17,ChalermsookCKLM20,ManurangsiRS21} can be written in this form. 

To simplify the presentation, we will just focus on two known results from \cite{Manurangsi17,ChalermsookCKLM20}, both of which can be easily verified\footnote{In both of \cite{Manurangsi17,ChalermsookCKLM20}, the graph $G$ is created explicitly and it is shown that, in the YES case, the graph contains a $k$-clique while, in the NO case, the graph does not contain a large biclique. (Note that, in \cite{Manurangsi17}, the NO case is even stronger as $n^{-o(1)}$-dense subgraphs are ruled out.)} to be of this form. 

First, under the Gap Exponential Time Hypothesis (Gap-ETH), the problem is hard to approximate to within a factor of $n^{o(1)}$, as stated more precisely below.

\begin{theorem}[\cite{Manurangsi17}] \label{thm:biclique-inapprox}
Assuming Gap-ETH, for any function $g$ such that $g = o(1)$, there is no polynomial-time algorithm for \gapbicliqueprob{k}{k/n^{g(n)}}.
\end{theorem}

For the parameterized regime, the problem is known to be hard to approximation to $o(k)$ factor under Gap-ETH (where $k$ is the parameter):

\begin{theorem}[\cite{ChalermsookCKLM20}] \label{thm:biclique-fpt-inapprox}
Assuming Gap-ETH, for any function $g$ such that $g = o(1)$, there is no FPT algorithm for \gapbicliqueprob{k}{k/g(k)}.
\end{theorem}

From the above results, we can apply our reductions to prove hardness of approximation results for all problems of interest, as specified in more detail below.

\paragraph{Ladder Index and Semi-Ladder Index.} \Cref{thm:ladder-inapprox} follows immediately from plugging in the reduction from \Cref{lem:biclique-to-sli} to \Cref{thm:biclique-inapprox}. Similarly, \Cref{thm:ladder-fpt-inapprox} follows immediately from plugging in the reduction from \Cref{lem:biclique-to-sli} to \Cref{thm:biclique-fpt-inapprox}.

\paragraph{Threshold Dimension.} \Cref{thm:td-inapprox} follows immediately from plugging in the reductions from \Cref{lem:biclique-to-sli} and \Cref{lem:red-sli-to-closuretd} to \Cref{thm:biclique-inapprox}. Similarly, \Cref{thm:td-fpt-inapprox} follows immediately from plugging in the reductions from \Cref{lem:biclique-to-sli} and \Cref{lem:red-sli-to-closuretd} to \Cref{thm:biclique-fpt-inapprox}.

\paragraph{Extended Threshold Dimension.} \Cref{thm:etd-inapprox} follows immediately from plugging in the reductions from \Cref{lem:biclique-to-sli}, \Cref{lem:red-sli-to-closuretd} and \Cref{lem:tdclosure-to-etd} to \Cref{thm:biclique-inapprox}. Similarly, \Cref{thm:etd-fpt-inapprox} follows from plugging in the reductions from \Cref{lem:biclique-to-sli}, \Cref{lem:red-sli-to-closuretd} and \Cref{lem:tdclosure-to-etd} to \Cref{thm:biclique-fpt-inapprox}.

%Before we discuss the hardness results, it is worth combining three reductions together in a single 

\section{co-NP-hardness of Extended Threshold Dimension}
\label{sec:conp}

Finally, we prove the co-NP-hardness of \etdprob (\Cref{thm:etd-co-np-hard}). To do this, we will reduce from a set splitting problem, as defined below:

\begin{center}
\fbox{
\begin{minipage}{0.95\textwidth}
\textbf{Problem:} \baltsetsplit \\
\textbf{Input:} Subsets $S_1, \dots, S_N \subseteq [M]$ each of size 4. \\
\textbf{Question:} Is $T \subseteq [M]$ of size $M/2$ such that $|S_i \cap T| = 2$ for all $i \in [N]$.
\end{minipage}
}
\end{center}

It is well known that this problem is NP-hard:

\begin{lemma}[\cite{Guruswami03}] \label{lem:balsetsplit-nphard}
\baltsetsplit is NP-hard.
\end{lemma}

We note that the above formulation is not exactly the same as stated in \cite{Guruswami03}, which does not contain the ``balancedness'' condition that $|T| = M/2$. Nevertheless, it is not hard to check that the reduction of \cite{Guruswami03} already satisfies balancedness. However, we opt to state the (simple) reduction in Appendix~\ref{app:balsetsplit} for completeness.

\begin{proof}[Proof of \Cref{thm:etd-co-np-hard}]
We reduce from \baltsetsplit to the complementary of \etdprob. Let $S_1, \dots, S_N \subseteq [M]$ be the input to \baltsetsplit. Let $\cX = [M]$ and we define our hypothesis class $\cH$ as follows. First, let $\cHbase$ be the class that contains the following functions: $\bzero, \bone_{x}$ for all $x \in [M]$, $\bone_{S_i}$ for all $i \in [N]$, and $\bone_{S_i \setminus \{x\}}$ for all $i \in [N]$ and $x \in S_i$. Then, let $\cH = \cHbase \cup (\cHbase)^{\bone}$. Finally, let $k = M/2 + 1$.

\paragraph{(Completeness)} Suppose that there exists $T \subseteq [M]$ of size $M/2$ such that $|S_i \cap T| = 2$. We claim that $\exthd{\cH} \leq k$; more specifically, $\thdclosure{\cH^{\bone_T}} \leq k$. To prove this, it suffices\footnote{This is simply because the generator $g_{d + 1}$ from \Cref{lem:closure_Threshold} always satisfies $\|g_{d+1}\|\geq d$.} to show that, for any $h \in \cH^{\bone_T}$, we have $\|h\| \leq k$. Notice that $\cH^{\bone_T} = \cHbase^{\bone_T} \cup \cHbase^{\bone_{\oT}}$. Thus, due to symmetry (between $T$ and $\oT$), it is in turn sufficient to prove that $\|h\| \leq k$ for all $h \in \cHbase^{\bone_T}$. Note that $h \in \cHbase^{\bone_T}$ is equal to $g \oplus \bone_T$ for some $g \in \cHbase$. We consider four cases:
\begin{itemize}
\item \textbf{Case I: } $g = \bzero$. We simply have $\|h\| = |T| = M/2$.
\item \textbf{Case II: } $g = \bone_x$. We simply have $\|h\| \leq |T| + 1 = M/2 + 1$.
\item \textbf{Case III: } $g = \bone_{S_i}$. Since $|S_i \cap T| = |S_i \cap \oT|$, we have $\|h\| = |T| = M/2$.
\item \textbf{Case IV: } $g = \bone_{S_i \setminus \{x\}}$. We have $\|h\| = \|(\bone_{S_i} \oplus \bone_T) \oplus \bone_x\| \leq \|\bone_{S_i} \oplus \bone_T\| + 1 \leq M/2 + 1$, where the last inequality follows from the previous case.
\end{itemize}
Thus, in all cases we have $\|h\| \leq M/2 + 1 = k$, which implies that $\exthd{\cH} \leq k$ as desired.

\paragraph{(Soundness)} Suppose contrapositively that $\exthd{\cH} \leq k = M/2 + 1$; that is, there exists $f \in \{0, 1\}^{\cX}$ such that $\thdclosure{\cH^f} \leq M/2 + 1$. By our definition of $\cH$, we have that $\cH^{f} = \cH^{\bone \oplus f}$. Thus, we may assume w.l.o.g. that $\|f\| \geq M/2$. Let $W = \supp(f)$, and $x_1, \dots, x_{\|f\|}$ be elements of $W$ (in arbitrary order). Consider the following cases:
\begin{itemize}
\item \textbf{Case I: } $\|f\| \geq M/2 + 2$. We claim that this case is impossible, i.e. $\thdclosure{\cH^f}$ must be at least $M/2 + 2$. To see this, consider the generator $g_{\|f\| + 1} = f \oplus \bzero$ and $g_i = f \oplus \bone_{x_i}$ for all $i \in [\|f\|]$. It is simple to verify that these satisfy the conditions in \Cref{lem:closure_Threshold}. Thus, we must have $\thdclosure{\cH^f} \geq \|f\| \geq M/2 + 2$.
\item \textbf{Case II: } $\|f\| = M/2 + 1$. Again, we claim that this case is impossible, i.e. $\thdclosure{\cH^f} \geq M/2 + 2$. To see this, additionally let $x_{\|f\| + 1}$ be any element of $[M] \setminus \supp(f)$. Then, let the generators be $g_{\|f\|+2} = f \oplus \bone_{x_{\|f\| + 1}}, g_{\|f\| + 1} = f \oplus \bzero$ and $g_i = f \oplus \bone_{x_i}$ for all $i \in [\|f\|]$. Again, these satisfy the conditions in \Cref{lem:closure_Threshold} and, thus, $\thdclosure{\cH^f} \geq \|f\| + 1 \geq M/2 + 2$.
\item \textbf{Case III: } $\|f\| = M/2$. We claim that $|S_i \cap W| = 2$ for all $i \in [N]$. Suppose for the sake of contradiction that $|S_i \cap W| \ne 2$ for some $i \in [N]$. We may assume w.l.o.g. that\footnote{Otherwise, we can consider $f \oplus 1$ instead of $f$, which ``flips'' $W$ to $\oW$.} $|S_i \cap W| \leq 1$. Let $x_1, \dots, x_{M/2 - 1}$ be distinct elements of $W \setminus S_i$ and $x_{M/2}, x_{M/2 + 1}, x_{M/2 + 2}$ be distinct elements of $S_i \setminus W$. Consider the following generator:
\begin{itemize}
\item $g_{M/2 + 3} = \bone_{S_i} \oplus f$,
\item $g_{j} = \bone_{S_i \setminus \{x_j\}} \oplus f$ for $j = M/2, M/2 + 1, M/2 + 2$, and,
\item $g_{\ell} = \bone_{x_\ell} \oplus f$ for $\ell = 1, \dots, M/2 - 1$.
\end{itemize}
Again, it is straightforward to verify that these satisfy the conditions in \Cref{lem:closure_Threshold}. This implies $\thdclosure{\cH^f} \geq \|f\| + 1 \geq M/2 + 2$, a contradiction.
\end{itemize}
Thus, we can conclude that $|W| = M/2$ and $|S_i \cap W| = 2$ for all $i \in [N]$ as desired.
\end{proof}

An interesting consequence of the proof above is that, if one can find $f^* = \argmin_{f} \thdclosure{\cH^f}$ in polynomial time, then P = NP. This is because such $f^*$ always corresponds to a solution for the \baltsetsplit problem. (Note that this consequence does not immediately follows from \Cref{thm:etd-co-np-hard}, which only implies the same result under the assumption NP $\ne$ co-NP.)

\begin{corollary}
Unless P = NP, there is no polynomial-time algorithm for computing $\argmin_{f} \thdclosure{\cH^f}$.
\end{corollary}
\section{Conclusion and Open Questions}
\label{sec:conclusion_and_open_questions}

In this work, we show computational hardness (of approximation) for the problem of computing the threshold dimension, the extended threshold dimension, the ladder index and the semi-ladder index. Given that the problem of computing the extended threshold dimension is both NP-hard and co-NP-hard, it remains an interesting question to prove a completeness result (e.g. with respect to the class $\Pi_2$) for this problem. Another interesting direction is to improve the hardness of approximation factor to $|\cX|^{1 - o(1)}$ and $|\cH|^{1 - o(1)}$ for (Extended) Threshold dimension, or to $n^{1 - o(1)}$ for (Semi-)Ladder Index. While $n^{1 - o(1)}$-factor inapproximability for Maximum Balanced Biclique is known~\cite{BhangaleGHKK16,Manurangsi17-icalp}, these hardness results are \emph{not} in the form of \tgapbicliqueprob used in our reduction. Therefore, we cannot directly apply our reduction to these results.

\paragraph{Acknowledgment.} I would like to thank Daniil Dmitriev and Amartya Sanyal for their helpful feedback on a previous version of this preprint.

\bibliography{ref}
\bibliographystyle{alpha}

\appendix
\section{NP-hardness of \baltsetsplit}
\label{app:balsetsplit}

In \cite{Guruswami03}, the NP-hardness result was actually stated for the following problem, which differs from \baltsetsplit in that it does not contain the condition $|T| = M/2$.

\begin{center}
\fbox{
\begin{minipage}{0.95\textwidth}
\textbf{Problem:} \tsetsplit \\
\textbf{Input:} Subsets $S_1, \dots, S_N \subseteq [M]$ each of size 4. \\
\textbf{Question:} Is $T \subseteq [M]$ such that $|S_i \cap T| = 2$ for all $i \in [N]$.
\end{minipage}
}
\end{center}

\begin{theorem}[\cite{Guruswami03}]
\tsetsplit is NP-hard.
\end{theorem}

We provide below a simple reduction from \tsetsplit to the variant we used in \Cref{sec:conp} (\baltsetsplit).

\begin{proof}[Proof of \Cref{lem:balsetsplit-nphard}]
We reduce from \tsetsplit. Let $(S_1, \dots, S_N)$ where $S_1, \dots, S_N \subseteq [M]$ denote an input instance of \tsetsplit. We construct an instance of \baltsetsplit as follows:
\begin{itemize}
\item Let the universe be $[M']$ for $M' = 2M$, where each $j \in [M]$ has two ``copies'' $j, j + M \in [M']$.
\item For every $i \in [N]$, create two subsets $S^{(0)}_i, S^{(1)}_i$ where $S^{(\ell)}_i = \{j + \ell M \mid j \in S_i\}$. In other words, create two copies of $S_i$ corresponding to each copy of the elements.
\end{itemize}
The output instance for \baltsetsplit is $S^{(0)}_1, S^{(1)}_1, \dots, S^{(0)}_N, S^{(1)}_N \subseteq [M']$. This reduction clearly runs in polynomial time. We now prove its completeness and soundness.

\paragraph{(Completeness)} Suppose that there exists $T \subseteq [M]$ such that $|S_i \cap T| = 2$ for all $i \in [N]$. Let $T' = T \cup \{j + M \mid j \in [M] \setminus T\}$. It is simple to see that $|T'| = M'/2$ and $|S^{(0)}_i \cap T'| = |S^{(1)}_i \cap T'| = 2$ for all $i \in [N]$.

\paragraph{(Soundness)} Suppose contrapositively that there exists $T' \subseteq [M']$ such that $|S^{(0)}_i \cap T'| = |S^{(1)}_i \cap T'| = 2$ for all $i \in [N]$. Let $T = T' \cap [M]$. We have $|S_i \cap T| = |S^{(0)}_i \cap T'| = 2$ for all $i \in [N]$.
\end{proof}
\section{NP-hardness of \exactbicliqueprob}
\label{app:biclique}

In this section, we prove the NP-hardness of \exactbicliqueprob. (\Cref{lem:np-exactbiclique})
The reduction here is essentially the same as that from \cite{Manurangsi17} but with partial assignment to only 2 variables per vertex (as opposed to $\Omega(\sqrt{n})$ variables as in \cite{Manurangsi17}).

\begin{proof}[Proof of \Cref{lem:np-exactbiclique}]
We reduce from 3-SAT. Let $\Psi$ be an input formula on variable set $X = \{x_1, \dots, x_n\}$. We construct the graph $G = (V, E)$ as follows:
\begin{itemize}
\item Let $k = \binom{n}{2}$.
\item Let $V$ be the set of partial assignments to 2 variables; each vertex of $V$ is $\{(x_{i_1}, b_{i_1}), (x_{i_2}, b_{i_2})\}$ where $i_1, i_2 \in [n]$ are distinct and $b_{i_1}, b_{i_2} \in \{0, 1\}$.
\item Add an edge between every pair of vertices $\{(x_{i_1}, b_{i_1}), (x_{i_2}, b_{i_2})\}$ and $\{(x_{i'_1}, b_{i'_1}), (x_{i'_2}, b_{i'_2})\}$ such that (i) the assignments are consistent (i.e. if $i_j = i'_{j'}$, then $b_{i_j} = b_{i'_{j'}}$) and (ii) every clause of $\Psi$ whose variables all belong to the set $\{x_{i_1}, x_{i_2}, x_{i'_1}, x_{i'_2}\}$ is satisfied by the (partial) assignment.
\end{itemize}
This reduction runs in polynomial time. We will next prove its completeness and soundness.

\paragraph{(Completeness)} If there exists a satisfying assignment $\phi: X \to \{0, 1\}$ for the formula $\Psi$, then we can simply let $S$ be the set of vertices $\{(x_{i_1}, \phi(x_{i_1})), (x_{i_2}, \phi(x_{i_2}))\}$ for all distinct $i_1, i_2 \in [n]$. It is clear that this forms a $k$-clique in $G$.  

\paragraph{(Soundness)} Suppose contrapositively that there exists a biclique in $B[G^o]$ with at least $k^2$ edges. Let this biclique be defined by $A \subseteq L$ and $B \subseteq R$. By the definition of $B[G^o]$, $A$ and $B$ correspond to sets of vertices in $V$, which we will denote as $U \subseteq V$ and $W \subseteq V$. The assumption that $A \times B$ has at least $k^2$ edges implies $|U| \cdot |W| \geq k^2 = \binom{n}{2}^2$. Furthermore, for every $u \in U$ and $w \in W$, either $u = w$ or the edge $(u, w)$ must exist in $G$.

To formalize the components of these vertices, we define a \textit{literal assignment} as a pair $(x_i, b) \in X \times \{0, 1\}$, representing assigning the truth value $b$ to the variable $x_i$. Every vertex in $V$ is a set of exactly two literal assignments for distinct variables. 

Let $S_U = \bigcup_{u \in U} u$ be the set of all literal assignments present across all vertices in $U$, and let $S_W = \bigcup_{w \in W} w$ be the corresponding set for $W$. 

%\textbf{Claim 1: Bounding the total number of unique literal assignments.}
For each variable $x_i \in X$, let $c_i \in \{0, 1, 2\}$ be the number of truth values assigned to $x_i$ in $S_U$, and let $d_i \in \{0, 1, 2\}$ be the number of truth values assigned to $x_i$ in $S_W$. 
Because $U \times W$ forms a biclique, every $u \in U$ must be consistent with every $w \in W$. Therefore, it is impossible for $S_U$ to contain $(x_i, b)$ while $S_W$ contains $(x_i, 1-b)$, as that would imply the existence of some $u \in U$ and $w \in W$ that contradict each other, preventing the edge $(u, w)$ from existing in $G$. 
This implies that if $c_i > 0$ and $d_i > 0$, they must agree on the exact same single truth value, meaning $c_i = d_i = 1$. Thus, we have $c_i + d_i \leq 2$. Summing over all $n$ variables yields
\begin{align} \label{eq:total-single-literal}
|S_U| + |S_W| = \sum_{i=1}^n c_i + \sum_{i=1}^n d_i \leq 2n.
\end{align}

Notice that $U \subseteq \binom{S_U}{2}$ and $W \subseteq \binom{S_W}{2}$. Thus, using our initial lower bound, we have
\begin{align} \label{eq:prod}
\binom{n}{2}^2 \leq |U| \cdot |W| \leq \binom{|S_U|}{2} \binom{|S_W|}{2}.
\end{align}
For any non-negative integers $x, y$ satisfying $x + y \leq 2n$, the product $\binom{x}{2} \binom{y}{2}$ is strictly maximized when $x = y = n$. Thus, both \eqref{eq:total-single-literal} and \eqref{eq:prod} must be equalities. For \eqref{eq:prod} to be an equality, we must have $U = \binom{S_U}{2}$ and $W = \binom{S_W}{2}$. The latter implies\footnote{Otherwise, if we assume w.l.o.g. that $c_i \geq 2$, then $(x_i, 0), (x_i, 1) \in S_U$ and thus $\{(x_i, 0), (x_i, 1)\} \in \binom{S_U}{2} \setminus U$.} that $c_i \leq 1, d_i \leq 1$ for all $i \in [n]$. Meanwhile, for \eqref{eq:total-single-literal} to be an equality, we must have $c_i + d_i = 2$. This implies that $c_i = d_i = 1$ for all $i \in [n]$. In other words, $S_U = S_W$, and this set contains exactly one literal assignment for every variable in $X$.  Let us call the corresponding assignment $\phi^*$.

Since $U = \binom{S_U}{2}$ and $W = \binom{S_W}{2}$, the sets $U$ and $W$ must contain every possible pair of literal assignments from $\phi^*$. 

Consider any clause in $\Psi$. Let the three variables involved in this clause be $x_p, x_q,$ and $x_r$. 
Consider the following two vertices:
\begin{align*}
u &= \{(x_p, \phi^*(x_p)), (x_q, \phi^*(x_q))\}, \\
w &= \{(x_p, \phi^*(x_p)), (x_r, \phi^*(x_r))\}.
\end{align*}
Since $U$ and $W$ contain all pairs from $\Psi$, we know $u \in U$ and $w \in W$. For $U \times W$ to be a biclique, the edge $(u, w)$ must exist in $G$. From our definition of $G$, this implies that $\phi^*$ satisfies this clause.

Thus, $\phi^*$ is a satisfying assignment for $\Psi$, which concludes our proof.
\end{proof}

\section{On Existence of Intersection-Closed Representation}

In this section, we study a related question of whether, for any class $\cH \subseteq \{0, 1\}^{\cX}$, it is always possible to find a shift $f \in \{0, 1\}^{\cX}$ such that the $f$-representation class $\cH^f$ is intersection-closed, i.e. $\overline{\cH^f} = \cH^f$. We show that this is false in general, but deciding its possibility admits a polynomial-time algorithm. This answers another question asked by \cite{DmitrievFHS26}. We remark that these same results were also independently discovered by Sanyal~\cite{sanyal2026private}.

\subsection{A Counterexample to Universal Existence}

\begin{theorem}
There exists a finite hypothesis class $\cH \subseteq \{0, 1\}^\cX$ such that for every shift $f: \cX \to \{0, 1\}$, the $f$-representation $\cH^f$ is not intersection-closed.
\end{theorem}
\begin{proof}
Let $\mathcal{X} = \{1, 2, 3\}$. %We identify subsets of $\mathcal{X}$ with their indicator functions as binary strings in $\{0, 1\}^3$. 
For convenience, we simply write functions on $\cX$ as binary strings in $\{0, 1\}^3$.
Define the hypothesis class $\mathcal{H}$ as the complete hypercube except the two antipodal points:
$$\mathcal{H} = \{0, 1\}^3 \setminus \{(0,0,0), (1,1,1)\}$$
For any arbitrary shift $f \in \{0, 1\}^3$, the shifted class $\mathcal{H}^f = \{h \oplus f \mid h \in \mathcal{H}\}$ preserves the structural property of missing exactly two antipodal points. Specifically, $\mathcal{H}^f = \{0, 1\}^3 \setminus \{a_f, b_f\}$, where the missing elements satisfy $a_f \oplus b_f = (1, 1, 1)$.
This limits $\{a_f, b_f\}$ to exactly four possible pairs. We show that in every case, we can find two elements $u, v \in \mathcal{H}^f$ whose intersection evaluates to one of the missing elements, thereby violating intersection-closure:
\begin{itemize}
    \item \textbf{Case 1: $\{a_f, b_f\} = \{(0,0,0), (1,1,1)\}$.} Let $u = (0,0,1)$ and $v = (0,1,0)$. Both belong to $\mathcal{H}^f$, but $u \wedge v = (0,0,0) \notin \mathcal{H}^f$.
    \item \textbf{Case 2: $\{a_f, b_f\} = \{(0,0,1), (1,1,0)\}$.} Let $u = (0,1,1)$ and $v = (1,0,1)$. Both belong to $\mathcal{H}^f$, but $u \wedge v = (0,0,1) \notin \mathcal{H}^f$.
    \item \textbf{Case 3: $\{a_f, b_f\} = \{(0,1,0), (1,0,1)\}$.} Let $u = (0,1,1)$ and $v = (1,1,0)$. Both belong to $\mathcal{H}^f$, but $u \wedge v = (0,1,0) \notin \mathcal{H}^f$.
    \item \textbf{Case 4: $\{a_f, b_f\} = \{(1,0,0), (0,1,1)\}$.} Let $u = (1,0,1)$ and $v = (1,1,0)$. Both belong to $\mathcal{H}^f$, but $u \wedge v = (1,0,0) \notin \mathcal{H}^f$.
\end{itemize}
Thus, for every possible $f \in \{0, 1\}^3$, the representation $\mathcal{H}^f$ is not intersection-closed.
\end{proof}

\subsection{Polynomial-Time Algorithm}

While the existence of an intersection-closed representation is not always guaranteed, finding one (when it exists) is computationally tractable, contrasting the hardness of computing the Extended Threshold dimension we proved earlier.

\begin{theorem} \label{thm:poly_time_closure}
Given a hypothesis class $\mathcal{H} \subseteq \{0, 1\}^\mathcal{X}$, there is a polynomial-time algorithm to determine whether there exists a shift $f \in \{0, 1\}^\mathcal{X}$ such that $\mathcal{H}^f$ is intersection-closed.
\end{theorem}

In fact, the above theorem is a simple consequence of the following lemma, which shows that if such a shift $f$ exists, then at least one such shift belongs to $\cH$.

\begin{lemma} \label{lem:shift_in_H}
If there exists a shift $f \in \{0, 1\}^\mathcal{X}$ such that $\mathcal{H}^f$ is intersection-closed, then there exists some $u \in \mathcal{H}$ such that $\mathcal{H}^u$ is intersection-closed.
\end{lemma}
\begin{proof}
Suppose there exists a shift $f \in \{0, 1\}^\mathcal{X}$ such that $\mathcal{H}^f$ is intersection-closed.
Since $\mathcal{H}^f$ is closed under intersection, it has a unique minimum element $m = \bigwedge_{x \in \mathcal{H}^f} x \in \mathcal{H}^f$.
Since $m \in \mathcal{H}^f$, we can write $m = u \oplus f$ for some $u \in \mathcal{H}$. We claim that $\mathcal{H}^u$ is also intersection-closed.

To see this, note that any element in $\mathcal{H}^u$ can be written as $h \oplus u = h \oplus m \oplus f = (h \oplus f) \oplus m$ for some $h \in \mathcal{H}$.
Since $h \oplus f \in \mathcal{H}^f$, we have $\mathcal{H}^u = \{ x \oplus m \mid x \in \mathcal{H}^f \}$.
Furthermore, because $m$ is the minimum element of $\mathcal{H}^f$, we have $m_i \le x_i$ for all $x \in \mathcal{H}^f$ and $i \in \mathcal{X}$, which implies $x \oplus m = x - m$.
Therefore, for any $x, y \in \mathcal{H}^f$, we have:
\[
(x \oplus m) \wedge (y \oplus m) = (x - m) \wedge (y - m) = (x \wedge y) - m = (x \wedge y) \oplus m.
\]
Since $\mathcal{H}^f$ is intersection-closed, we have $x \wedge y \in \mathcal{H}^f$, which implies $(x \wedge y) \oplus m \in \mathcal{H}^u$.
Thus, $\mathcal{H}^u$ is intersection-closed.
\end{proof}

\Cref{thm:poly_time_closure} now follows immediately.

\begin{proof}[Proof of \Cref{thm:poly_time_closure}]
By \Cref{lem:shift_in_H}, our algorithm can proceed as follows: Iterate through each $u \in \mathcal{H}$, and check whether $\cH^u$ is intersection-closed. Note that the check can be performed in polynomial-time by iterating through all $a, b \in \cH^u$ and check whether $a \wedge b$ belongs to $\cH^u$.
\end{proof}

% ENDALLOWEDIT

\end{document}